\documentclass[reprint,
superscriptaddress,
 amsmath,amssymb
 aps,
longbibliography,
prl
]{revtex4-2}

\usepackage{xparse}
\usepackage{amsmath, amssymb, amsthm}

\usepackage{appendix}
\usepackage{graphicx}
\graphicspath{ {./figures}}

\usepackage{dcolumn}
\usepackage{bm}
\usepackage{hyperref}
\usepackage{xcolor}

\usepackage{comment}
\usepackage{dsfont}

\usepackage{upref} 
\usepackage{subdepth}
\usepackage{enumerate}
\usepackage{times}

\hypersetup{colorlinks,citecolor=blue,filecolor=blue,linkcolor=blue,urlcolor=navyblue}
\definecolor{navyblue}{rgb}{0.0, 0.0, 0.5}

\newtheorem{ex*}{Example}

\newtheorem{claim}{Claim}
\newtheorem{claim*}{Claim}

\begin{document}

\title{Geometric Delocalization in Two Dimensions}

\author{Laura Shou}
    \affiliation{Joint Quantum Institute, Department of Physics, University of Maryland, College Park 20742}
    
\author{Alireza Parhizkar}
    \affiliation{Joint Quantum Institute, Department of Physics, University of Maryland, College Park 20742}

\author{Victor Galitski}
    \affiliation{Joint Quantum Institute, Department of Physics, University of Maryland, College Park 20742}

\begin{abstract}
We demonstrate the existence of transient two-dimensional surfaces where a random-walking particle escapes to infinity in contrast to localization in standard flat 2D space.  We first prove that any rotationally symmetric 2D membrane embedded in flat 3D space cannot be transient. Then we formulate a criterion for the transience of a general asymmetric 2D membrane. We use it to explicitly construct a class of transient 2D manifolds with a non-trivial metric and height function but ``zero average curvature,'' which we dub  tablecloth manifolds. The absence of the logarithmic infrared divergence of the Laplace--Beltrami operator in turn implies the absence of weak localization, non-existence of bound states in shallow potentials, and breakdown of the Mermin--Wagner theorem and Kosterlitz--Thouless transition on the tablecloth manifolds, which may be realizable in both quantum simulators and corrugated two-dimensional materials.
\end{abstract}

\maketitle

\textit{Introduction}---
There are many seemingly disconnected physical phenomena which are related to the properties of the Laplace--Beltrami operator. They include random walk or diffusion, the standard Schr{\"o}dinger equation, the properties of fluctuations in symmetry broken phases, interactions between topological excitations, and many others. In particular, if the  heat kernel $p(x,y;t)$ integrated over time --- the Green's function of the Laplace--Beltrami operator on a manifold --- is infinite it implies automatically the following properties of this space: a random-walking particle is guaranteed to return to its starting region infinitely often (recurrence) \cite{bull}, there is Anderson localization in an arbitrarily weak disorder potential \cite{AltshulerAranov1985,LeeRamakrishnan1985,CNG2025}, any shallow quantum potential well hosts a bound state \cite{ShallowPotential}, no long-range order with spontaneously broken continuous symmetry can exist at finite temperature in this space (Mermin--Wagner theorem) \cite{MerminWagner1966,Hohenberg1967,FriedliVelenik}, to name a few. Specifically in flat Euclidean space, the two-dimensional case is critical 
as the corresponding heat kernel integral diverges logarithmically:
\begin{align}
\label{log}
 \int\limits_{\tau_{\rm min}}^{\tau_{\rm max}} p(x,y;t)\,dt \propto \ln (\tau_{\rm min} /\tau_{\rm max}) \, ,
\end{align}
which means that the two-dimensional flat space is recurrent, while higher dimensional flat spaces, which have finite integrals, are transient (a random-walking particle always escapes its starting region). An additional physical phenomenon tied to the Green's function of the Laplace--Beltrami operator specific to $O(2)$ models on two-dimensional manifolds is the behavior of topological excitations there. In the conventional flat space, the logarithmic divergence \eqref{log} is tied to logarithmic vortex-vortex interactions that in turn lead to a finite-temperature Berezinskii--Kosterlitz--Thouless transition due to the competition with entropic effects which also scale logarithmically \cite{BKTGreenRef}. 

Due to the abundance of fundamental physical phenomena tied to the Laplace--Beltrami operator, it is reasonable to ask if there exist two-dimensional curved transient manifolds. 
Recent works \cite{CNG2025,chen2024anderson} have studied some of these phenomena on hyperbolic manifolds and lattices, which are known to be transient \cite{McKean1970,Prat1971,Kesten1959full,DerriennicGuivarch1973}.
Here we ask: is it possible to have a two-dimensional membrane (i.e., a two-dimensional smooth manifold embedded in flat Euclidean three-dimensional space described by a height function) that is transient, 
$\int_{\tau_\mathrm{min}}^\infty p(x,y;t)\,dt<\infty$?
We answer this question in the affirmative and explicitly construct 2D transient ``tablecloth manifolds'' (Fig.~\ref{fig:Manifolds}), which hence lead to the breakdown of standard two-dimensional physics. 

\begin{figure}[htb]
\includegraphics[width=\linewidth]{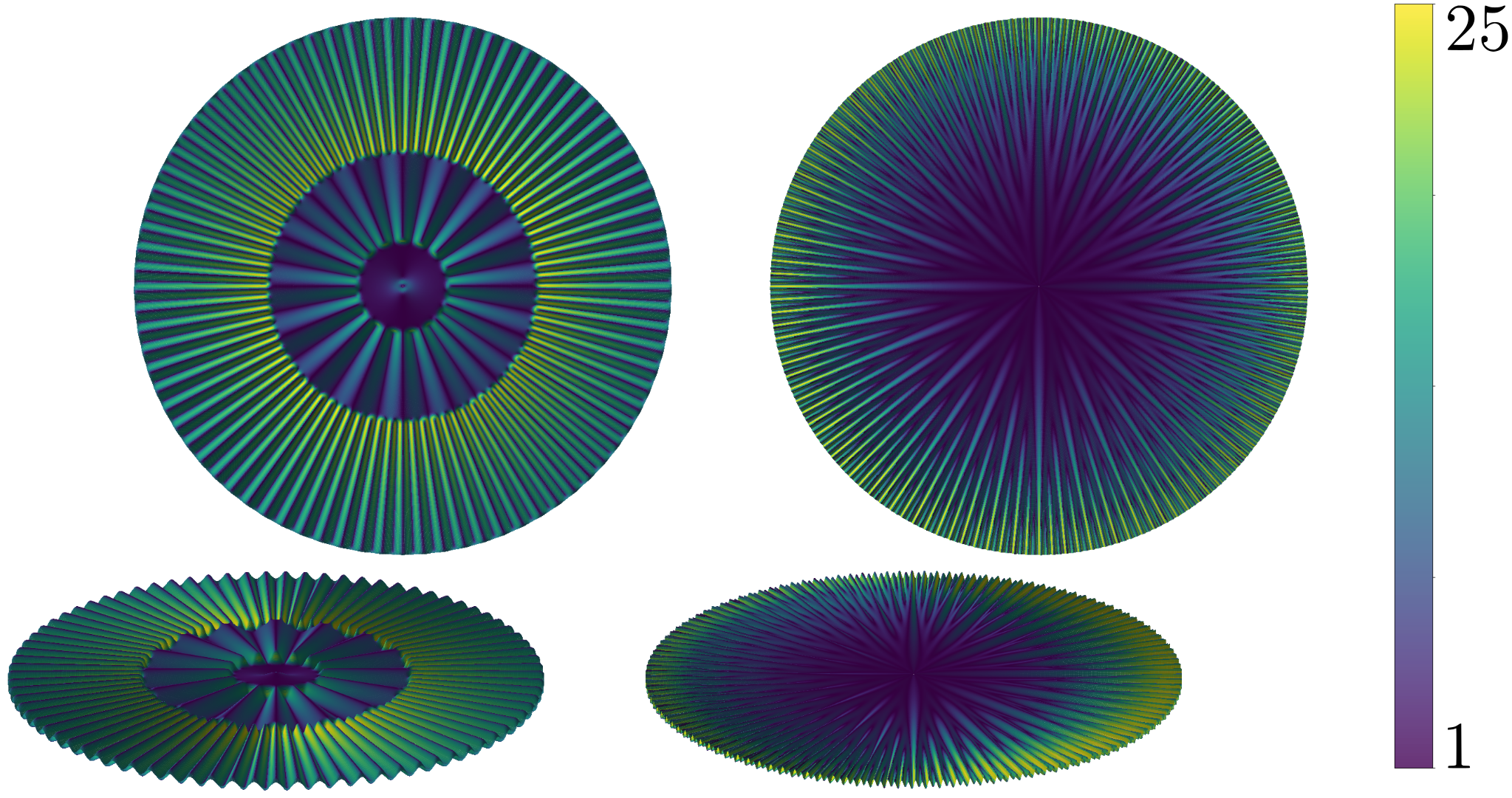}
\centering
\caption{Two examples of tablecloth manifolds. On the right, a generic tablecloth manifold, and on the left the simplified example constructed in Eq.~\eqref{eqn:h}. The colors demonstrate the ratio of the volume element at any point on the manifold, $\sqrt{g(r,\theta)}drd\theta$, to the regular flat volume element at the same point, $r dr d\theta$. Therefore, the colors encode the value of $\sqrt{g(r,\theta)}/r$. This allows us to compare the volume growth to the regular flat one as we go away from the origin with increasing $r$. Purple corresponds to the regular $\pi r^2$ volume growth of a flat disk, while other colors designate faster growths.}
\label{fig:Manifolds}
\end{figure}

\textit{Transience on a 2-dimensional membrane}---
The general form of a two dimensional Riemannian metric in polar-like coordinates, with no isometries assumed, is given by,
\begin{equation}\label{eqn:metric-polar-main}
    ds^2 = A^2(r,\theta)\, dr^2 + 2B(r,\theta) \,dr\, d\theta + C^2(r,\theta) \,d\theta^2 \, .
\end{equation}
For rotationally symmetric manifolds with metric of the form $ds^2=dr^2+f(r)^2\,d\theta^2$, a well-known result \cite{Ahlfors1935,Milnor1977} is that transience is equivalent to the condition
\begin{align}\label{eqn:rot-transience}
\int_1^\infty\frac{1}{f(r)}\,dr<\infty \, .
\end{align}
Using this, we can find rotationally symmetric metrics $ds^2=dr^2+f(r)^2\,d\theta^2$ corresponding to transient Brownian motion, and which also satisfy an average zero curvature condition, as described in the Supplemental Material \cite{sm}. However it will turn out none of these examples can represent a two-dimensional membrane in standard flat space.
While such metrics can be engineered in circuit-based platforms~\cite{kollar2019hyperbolic,boettcher2020quantum,blais2021circuit,lenggenhager2022simulating}, our objective here is to materialize the geometry directly as a curved physical surface embedded in flat three-dimensional space. This approach enables realization not only in synthetic architectures but also in naturally occurring systems involving corrugated membranes~\cite{BowickGiomi2d}.
As we will show, this embedding will necessitate breaking rotational symmetry.

To construct a transient surface described by a height function embedded in $\mathbf{R}^3$, we use the cylindrical coordinates, $x^\mu \equiv (r,\theta,z)$, on the three dimensional ambient manifold and polar-like coordinates, $x^a \equiv (r,\theta)$, on the embedded membrane where we have identified the $r$ and $\theta$ coordinates of both manifolds. This way $\xi^\mu \equiv (r,\theta,h(r,\theta))$ is the vector that probes the membrane with $h(r,\theta)$ being its height function.

The induced metric on the membrane is then given by,
\begin{equation}
    g_{ab} = \frac{\partial \xi^\mu}{\partial x^a} \frac{\partial \xi^\nu}{\partial x^b} \eta_{\mu\nu} \, ,
\end{equation}
with the Einstein summation rule presumed and $\eta_{\mu\nu}$ being the flat metric of the ambient space with cylindrical coordinates,
\begin{equation}
    \eta_{\mu\nu} = 
        \begin{bmatrix}
            1 & 0 & 0\\
            0 & r^2 & 0 \\
            0 & 0 & 1
        \end{bmatrix}	\, .
\end{equation}
By noticing that $\partial r / \partial r = \partial \theta /\partial \theta = 1$ and $\partial r/ \partial \theta = 0$, we calculate the induced metric to be,
\begin{equation} \label{eqn:gI}
    g_{ab} = 
        \begin{bmatrix}
            1+ \left(\frac{\partial h}{\partial r}\right)^2 & \frac{\partial h}{\partial r} \frac{\partial h}{\partial \theta} \\
            \frac{\partial h}{\partial \theta} \frac{\partial h}{\partial r} & r^2 +  \left(\frac{\partial h}{\partial \theta}\right)^2
        \end{bmatrix}	\, ,
\end{equation}
with the corresponding scalar curvature given by,
\begin{equation}
    R = \frac{2}{r}\frac{ \left(\frac{1}{r}\partial^2_\theta h +  \partial_r h\right) \partial^2_r h -  \frac{1}{r} \left(\frac{1}{r}\partial_\theta h -  \partial_r \partial_\theta h\right)^2}{ \left[1  +  (\partial_r h)^2 + (\frac{1}{r}\partial_\theta h)^2\right]^2} \, .
\end{equation}
If we demand rotational symmetry around the origin, then the induced metric turns into
\begin{equation}\label{eqn:gI-rot}
    g_{ab} = 
        \begin{bmatrix}
            1+ \left(\frac{\partial h}{\partial r}\right)^2 & 0 \\
            0 & r^2
        \end{bmatrix}	\, ,
\end{equation}
which is of the form, $ds^2 = A^2(r) dr^2 + r^2 d\theta^2$, with $A(r)=\sqrt{1+ h'(r)^2}$.~\footnote{The non-zero Christoffel symbols, $\Gamma^a_{bc} \equiv \frac{1}{2}g^{ad}( \partial_b g_{cd} + \partial_c g_{bd} - \partial_d g_{bc})$, for this metric are:$    \Gamma^r_{\theta\theta} = \frac{r}{A^2(r)} \, , \quad \Gamma^\theta_{r\theta} = \frac{1}{r} \quad  \text{and} \quad \Gamma^r_{rr} = \frac{A'(r)}{A(r)} \, ,$
which leads to the only independent component of the Riemann tensor, $R^r_{\ \theta r\theta}= r A'/A^3$, and the scalar curvature,
\begin{equation}
    R = 2g^{\theta\theta} R^r_{\ \theta r\theta} =  \frac{2}{r}\frac{A'(r)}{ A^3(r)} \, .
\end{equation}
In two dimensions, this Ricci scalar is twice the Gaussian curvature. Hence, for a spherically curved membrane with $h(r)=\sqrt{1 - r^2}$, $R=2$, while the Gaussian curvature is the inverse product of the principal radii of the membrane, which here is equal to one everywhere.}

However, having the line element as $ds^2 = A^2(r) dr^2 + r^2 d\theta^2$ means that the physical distance between two neighboring points that are radially apart is $dr_{\text{ph}} \equiv A(r) dr$. So in terms of the physical distance we have $ds^2 = dr_{\text{ph}}^2 + r(r_{\text{ph}})^2 d\theta^2$. Since our metric is now in the form required to apply Eq.~\eqref{eqn:rot-transience}, we see the transience condition becomes,
\begin{equation}
    \int_1^\infty \frac{dr_{\text{ph}}}{r(r_{\text{ph}})} <\infty \, .
\end{equation}
But this is only possible if $A(r)<1$ since otherwise $dr_\text{ph} > dr$ (or $ r_\text{ph} =\int^r A(r')dr' \ge r$) and thus $\int_1^\infty dr_{\text{ph}}/r(r_{\text{ph}}) \ge \int_1^\infty dr/r = \infty$. 
But here $A(r)=\sqrt{1 + h'(r)^2}>1$. This means in order to have transience on a membrane we need to break rotational symmetry. Transience on a membrane is unachievable with rotational symmetry.

Now considering non-rotationally symmetric manifolds, a useful condition for transience was given in Ref.~\cite{Doyle1988}. Using a heuristic argument involving conductance (which can be made rigorous, see e.g. Ref.~\cite[\S12]{bull}), Ref.~\cite{Doyle1988} showed that for a two dimensional manifold with a metric given by $ds^2=dr^2+f(r,\theta)^2\,d\theta^2$, a sufficient condition for transience is
\begin{align}\label{eqn:geo-transience}
\mathrm{meas}\left\{\theta\in[0,2\pi):\int_1^\infty\frac{dr}{{f(r,\theta)}}<\infty\right\}>0 \, ,
\end{align}
where $\mathrm{meas}$ means the Lebesgue measure. 
Intuitively, this means the manifold only needs to look transient on e.g. a small wedge of $\theta$ in order to have the Brownian motion escape. 
The argument given in Ref.~\cite{Doyle1988} is that the resistance of a piece is proportional to its length divided by its cross-sectional area, and so the resistance along a narrow strip of arc width $\delta\theta$ out to infinity is given by integrating $1/(\sqrt{g(r,\theta)}\delta\theta)=1/(f(r,\theta)\delta\theta)$ with respect to $r$, with $g(r,\theta) \equiv \det(g_{ab})$ being the determinant of the metric. 
By cutting the manifold into many small strips, each going from a center ring to infinity, the conductances of each strip, which are the reciprocals of the resistances, add to give the conductance of the whole system. As long as this is nonzero, this implies a current flow out to infinity, which causes a Brownian motion to escape.

In principle, we can transform a metric of the form Eq.~\eqref{eqn:metric-polar-main}
into the form $ds^2=dr^2+f(r,\theta)^2\,d\theta^2$, and then the condition Eq.~\eqref{eqn:geo-transience} can be applied.
However, in practice, the required change of coordinates can be very difficult and not analytically tractable. 
Therefore in order to prove transience for a membrane case described by Eq.~\eqref{eqn:gI}, we need to derive the analogue of Eq.~\eqref{eqn:geo-transience} for a metric of the general form in Eq.~\eqref{eqn:metric-polar-main}.
We first note that the analogue cannot simply be integrating $\int_1^\infty dr/\sqrt{ g(r,\theta)}$ and checking if it is finite for enough $\theta$.
For instance, while ${\sqrt{ g(r,\theta)}} dr d\theta$ gives the volume element at $(r,\theta)$ (as $f(r,\theta)drd\theta$ did in the metric $dr^2+f(r,\theta)^2\,d\theta^2$), the integration done to calculate the resistance should be done with respect to the \emph{physical} distance on the manifold, which is not necessarily the coordinate distance $r$ in Eq.~\eqref{eqn:metric-polar-main}.
As an example, consider the paraboloid given by $h(r,\theta)=\frac{1}{2}r^2$, which has positive curvature and recurrent Brownian motion. Using Eq.~\eqref{eqn:gI}, the induced metric is $ds^2=(1+r^2)\,dr^2+r^2\,d\theta^2$, which has  $\sqrt{ g(r)}=r\sqrt{1+r^2}\sim r^2$ for large $r$. Therefore $\int_1^\infty dr/\sqrt{ g(r)}<\infty$, but the problem is that the physical distance of a point $(r,\theta)$ on the manifold is much larger than the coordinate $r$, so that the integral over $dr$ does not represent anything meaningful.
Instead, we will show:
\begin{claim}\label{claim}
For a manifold with smooth metric $g_{ab}=\begin{bmatrix}A^2(r,\theta)&B(r,\theta)\\B(r,\theta)&C^2(r,\theta)\end{bmatrix}$ in polar-like coordinates with $g(r,\theta)>0$ for all $r>0$, the condition
\begin{align}\label{eqn:transience-cond}
\operatorname{meas}\left\{\theta\in[0,2\pi):\int_1^\infty\frac{A^2(r,\theta)}{\sqrt{g(r,\theta)}}\,dr<\infty\right\}>0
\end{align}
implies transience.
\end{claim}
We can explain this heuristically, similarly to the argument given in Ref.~\cite{Doyle1988} for Eq.~\eqref{eqn:geo-transience}. 
As described in Fig~\ref{fig:wedge}, the resistance of a thin wedge of angle $\delta\theta$ is proportional to $\int_{r_0}^\infty dr A^2(r,\theta)/\sqrt{g(r,\theta)}\,\delta\theta$.
To determine the conductance of a ball of radius $r_0$ out to infinity, we sum the reciprocals of the resistance of each strip over $\delta\theta$. This gives the total conductance as $\int_0^{2\pi}\,d\theta \left(\int_{r_0}^\infty dr A^2(r,\theta)/\sqrt{g(r,\theta)}\right)^{-1}$, which is nonzero exactly when Eq.~\eqref{eqn:transience-cond} is satisfied.
We also give the formal proof of Claim~\ref{claim} in the Supplemental Material \cite{sm} employing the concept of capacity.

\begin{figure}[htb]
\includegraphics{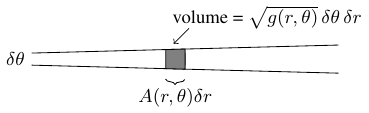}
\caption{Calculating the resistance along a thin strip with angle $\delta\theta$. The volume (which in two dimensions is the surface area) of the shaded region,$\sqrt{g(r,\theta)}\,\delta\theta\,\delta r$, is given by the product of its physical length $A(r,\theta)\delta r$ and its physical cross-section which therefore is given by $\sqrt{g(r,\theta)}\delta\theta \delta r/A(r,\theta)\delta r$. The resistance of the shaded region is proportional to its physical length divided by its cross-section.
Consequently, the resistance of the whole strip is proportional to $\int_{r_0}^\infty dr A^2(r,\theta)/\sqrt{g(r,\theta)}\,\delta\theta$.}\label{fig:wedge}
\end{figure}

Returning to the paraboloid example $h(r,\theta)=\frac{1}{2}r^2$, we see the integral in Eq.~\eqref{eqn:transience-cond} diverges (for all $\theta$), resolving the observations noted above. We also recover the fact that every rotationally invariant metric of the form in Eq.~\eqref{eqn:gI} must be recurrent, since
$\partial_\theta h=0$, $C(r)=r$, and $A(r)\ge1$, so the integral is $\int_1^\infty dr A(r)/C(r) \ge \int_1^\infty dr/r=\infty$.

Compared to the constant curvature hyperbolic spaces considered in Refs.~\cite{CNG2025,chen2024anderson}, 
Claim~\ref{claim} allows for non-homogeneous spaces with rapidly changing curvature and volume growth rates. This will be necessary for constructing a transient membrane in Euclidean 3D space and for obtaining an average 2D flatness, or zero curvature, property.

\textit{Tablecloth manifolds}---
A theorem of Ref.~\cite{ChengYau1975} states that any manifold with geodesic ball volume growth $O(r^2)$ will have recurrent Brownian motion. Therefore, for transience, the first obstacle to overcome is having sufficiently fast volume growth. 
We emphasize however that fast volume growth alone is \emph{not} sufficient for transience. This can be seen from e.g. Eq.~\eqref{eqn:rot-transience}, even for rotationally symmetric metrics, by mixing regions of slow and fast volume growth together in certain ways. Later in this section, we also give a membrane example with fast $r^5$ volume growth which fails to satisfy Eq.~\eqref{eqn:transience-cond}.
Worth mentioning is that the tablecloth membranes we construct are very non-homogeneous, involving regions of both fast and slow volume growth, and so the traditional scaling \cite{aalr} and $2+\varepsilon$ dimensionality \cite{wegner1989four} arguments for e.g. Anderson localization in homogeneous spaces, including for $\mathbf{R}^d$ and (bi)fractals \cite{schreiber1996dimensionality}, do not apply here.

In order to gain some intuition about the relation between volume growth and geometry consider the following examples.
On a unit sphere, as one gets away from a pole, the boundary of the ball centering the pole increases for a while and then shrinks as one passes through the equator. This means that the volume growth is slower than $c r^2$ for $0<r<2\pi$. The spherical metric is given by $ds^2 = dr^2 + \sin^2(r) d\theta^2$ and the boundary of the ball at $r=R$ by $\int_{r=R} ds = 2\pi\sin(R)$ which should be compared to $2\pi R$.
On hyperbolic geometry the reverse is true. Namely the boundary of the ball grows faster than the circumference of a regular circle.
The hyperbolic metric is given by $ds^2 = dr^2 + \sinh^2(r) d\theta^2$ and the boundary of the ball is given by $\int_{r=R} ds = 2\pi\sinh(R)$ which should be compared to $2\pi R$. But due to the same fact the maximally symmetric hyperbolic plane, described by the given metric, is not isometrically embeddable in the three dimensional flat space as a membrane. However, for fast volume growth with respect to $r$ we demand the boundary of the ball $\int_{r=R} ds$ to increase faster that the regular circle. Therefore, there is no other way for the membrane but to wrinkle up. This is again why we need to drop rotational symmetry. Fig.~\ref{fig:Tablecloth} demonstrates this concept.

\begin{figure}[htb]
\includegraphics[width=\linewidth]{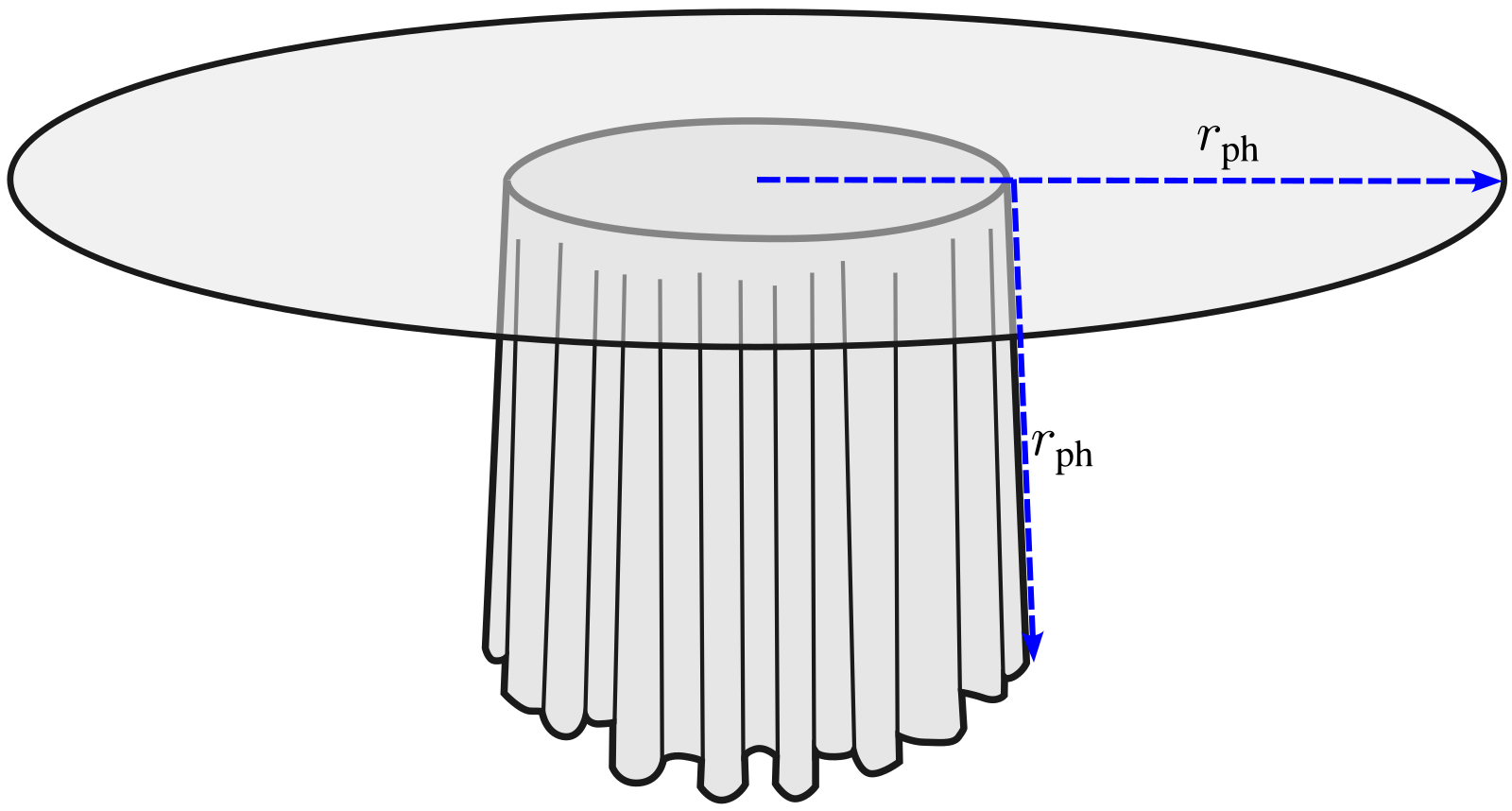}
\centering
\caption{A disk shaped tablecloth has no wrinkles when lying flat on a flat two dimensional surface. When the tablecloth is set on a circular table and drapes from it, it needs to fit into a new geometry, transitioning from disk $ds_D^2 = dr_{\text{ph}}^2 + r_{\text{ph}}^2 d\theta^2$ to cylinder $ds_C^2 = dr_{\text{ph}}^2 + R^2 d\theta$ with $R$ being the radius of the circular table. Because of the volume mismatch, $\sqrt{g_D}/\sqrt{g_C}=r_{\text{ph}}/R$, the tablecloth needs to wrinkle up. Also note how the radial coordinate, $r$, differs from $r_{\text{ph}}$. $r_{\text{ph}}$ is the physical distance a traveler takes on the tablecloth, while $r$ is identified with the radial component of the cylindrical coordinates.}
\label{fig:Tablecloth}
\end{figure}

We also note the connection between curvature and transience or recurrence of Brownian motion. Intuitively, negative curvature regions (like in hyperbolic space) cause a Brownian traveler to leave compact regions quickly, while zero and positive curvature regions do not. Therefore we expect that 2D surfaces with sufficiently negative curvature at infinity will have transient Brownian motion, while those with zero or positive curvature will have recurrent Brownian motion; for precise statements see Refs.~\cite{Milnor1977,Doyle1988}. However, the `tablecloth manifolds' we construct have a mix of negative and positive curvature regions, giving an average zero curvature as described below.
Despite this average ``flatness'', the negative curvature regions can still win out and cause the Brownian motion to escape.
This resonates with the ideas of Zeldovich considering a ``universe homogeneous in the mean'' with randomly fluctuating metric that still appears hyperbolic \cite{zel1964observations,sokoloff2015intermittency}.
There, Zeldovich argued using the idea of ``intermittency'' that an overall flat universe with patches of positive and negative curvature would appear hyperbolic to an observer, due to exponential growth of geodesics in the negative curvature regions.

A generic tablecloth manifold is depicted in the right half of Fig.~\ref{fig:Manifolds}. For our considerations we would like to investigate a simplified example of a tablecloth manifold (see the left membrane in Fig.~\ref{fig:Manifolds}) given by the following height function,
\begin{align}\label{eqn:h}
h(r,\theta)&=\sum_{n=1}^\infty \chi_n(r)\cos(n^4\theta) \, ,
\end{align}
where $\chi_n(r)$ are smooth bump functions extended from $n(n-1)/2$ to $n(n+1)/2$, which also means each bump is centered at $r = n^2/2$ and has increasing length $\sim n$.
The $\chi_n$ are exactly equal to 1 for most of the interval, and exactly equal to zero in a small interval near each endpoint.
Examples of the bump functions are shown in Fig.~\ref{fig:bumpfunctions}, and precise requirements and a sample formula are given in the Supplemental Material \cite{sm}. 
We note that by construction of the $\chi_n(r)$, which only extend from $n(n-1)/2$ to $n(n+1)/2$, only one term in the sum in Eq.~\eqref{eqn:h} can contribute at any given $r$.

\begin{figure}[htb]

	  
	  
  
\includegraphics{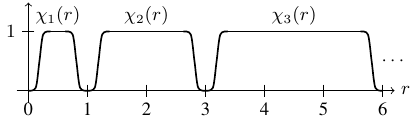}
\caption{Drawing of example bump functions $\chi_n(r)$, $n=1,2,3$. Each function $\chi_n$ is supported in $[n(n-1)/2,n(n+1)/2]$, and has a flat plateau within that interval where it is exactly equal to one, along with a small interval near the endpoints where it is exactly equal to zero. Precise requirements for the bump functions are given in the Supplemental Material \cite{sm}.}\label{fig:bumpfunctions}
\end{figure}

Within each annulus determined by a bump $\chi_n$, the membrane oscillates rapidly according to  $\cos(n^4\theta)$, with faster oscillations as $r$ increases. This causes the volume element $\sqrt{ g(r,\theta)}$ to increase rapidly, while the coordinate distance $r$ stays proportional to the physical distance on the membrane.
The membrane also stays bounded in height between $-1$ and $1$.

With the use of these bump functions, we can quickly see the membrane satisfies an average curvature zero condition: on the balls $B_R$ of radius $R=n(n+1)/2$ centered at the origin, the total Gaussian curvature is zero by the Gauss--Bonnet theorem, since the manifold is completely flat in the neighborhood of $\partial B_R$ where $\chi_n(r),\chi_{n+1}(r)=0$:
\begin{align}
\int_{B_R}K\,d\mathrm{Vol}_g&=2\pi\chi(B_R)-\int_{\partial B_R}k_g\,ds=2\pi-\frac{2\pi R}{R}=0 \, .
\end{align}
On these balls, the boundary $\partial B_R$ is particularly nice since it is flat, and does not grow rapidly like in hyperbolic space. 

For $\frac{n(n-1)}{2}+\frac{1}{4}<r<\frac{n(n+1)}{2}-\frac{1}{4}$, since $\chi_n\equiv1$ there, we have a simple form of the metric,
\begin{align}
g_{ab}&=\begin{bmatrix}
1& 0\\
0&r^2+n^8\sin^2(n^4\theta)
\end{bmatrix} \, .
\end{align}
Outside these regions, the metric is messier (involving the non-constant terms $A^2(r,\theta)$ and $B(r,\theta)$), but we still have $g(r,\theta)\ge r^2$ for a membrane.
As shown in the Supplemental Material, splitting up the the integral in Eq.~\eqref{eqn:transience-cond} into regions where $\chi_n\equiv1$ or not, and using that ${1+|\chi_n'(r)|^2}\le c'$ for some constant $c'$ and all $n$ and $r$, we can 
reduce the condition of Eq.~\eqref{eqn:transience-cond} to showing that
\begin{align}\label{eqn:series-main}
\sum_{n=1}^\infty\frac{1}{n\sqrt{1+n^4\sin^2(n^4\theta)}}<\infty \, .
\end{align}
The terms in this series tend to decay like $1/n^3$, except when $\sin^2(n^4\theta)$ is small, which is when $n^4\theta$ is close to a multiple of $\pi$. For these $n$ the terms behave like $1/n$, which is potentially problematic since $\sum_{n=1}^\infty 1/n$ diverges. However, we can quantify when $\sin^2(n^4\theta)$ is ``small'' using equidistribution of $\{n^4x\}$ modulo 1 for almost every (a.e.) $x$; more precisely, using the Erd\H{o}s--Tur\'an inequality \cite{ErdosTuran1948}, which gives a quantitative bound on the rate in Weyl's equidistribution theorem,
and the application by Ref.~\cite{Baker1981}.
This will imply that Eq.~\eqref{eqn:series-main} holds for a.e. $\theta$, and so the condition Eq.~\eqref{eqn:transience-cond} for transience is met.
We provide full details in the Supplemental Material \cite{sm}.

Interestingly enough, exactly due to the periodicity in $\theta$, we see that transience on tablecloth manifolds is connected to the convergence problem of series. For example, by considering the tablecloth manifold on a cylinder rather than a plane, the transience condition will be related to the convergence of series of the following type,
\begin{equation}
    \sum^\infty_{n=1} \frac{1}{\sqrt{1+n^p \sin^q(n^r \theta)}} \, . 
\end{equation}

Finally, we note there are several subtleties involved in the described membrane construction. In particular, it is possible to construct a tablecloth-like manifold, even with faster volume growth than that of Eq.~\eqref{eqn:h}, that does not satisfy Eq.~\eqref{eqn:transience-cond}.
For example, if we change $\chi_n(r)$ to have plateau length $1$ (Fig.~\ref{fig:even-bumps}) instead of length $n$, this {increases} the volume growth rate of the manifold, but causes the integral in Eq.~\eqref{eqn:transience-cond} to be divergent for every $\theta\in[0,2\pi)$. We expect, due to the fact that the radial distance factor $A^2(r,\theta)$ in the metric is uniformly bounded in $(r,\theta)$, that the arguments  \cite{Ahlfors1935,Nevanlinna1940,LyonsSullivan1984,Grigoryan1985existence} (see also a heuristic resistance argument in Ref.~\cite{Doyle1988}) actually imply recurrence. 
We expect this because the thin flat rings where $\chi_n(r)\equiv 0$ and the metric is flat, $ds^2=dr^2+r^2\,d\theta^2$, now occur linearly with $r$, and integrating $1/r$ over evenly spaced intervals leads to a divergent integral.
For generic tablecloth manifolds, where we do not impose strong cutoff functions like $\chi_n$, we do not even have any thin flat rings, and so we expect one would generally have transience.

\begin{figure}[htb]

	  
	  
  
\includegraphics{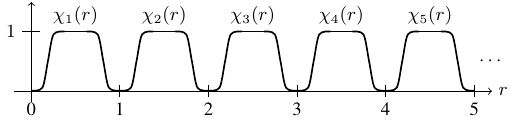}
\caption{Evenly spaced plateaus, leading to a specific example that fails to satisfy Eq.~\eqref{eqn:transience-cond}.}\label{fig:even-bumps}
\end{figure}

\textit{Conclusion}--- In summary, this paper demonstrates the existence of transient two-dimensional manifolds, which can be smoothly embedded as membranes in  flat three-dimensional space. The transient tablecloth manifolds constructed here are just one example of membranes with zero average curvature, and many other such transient geometries exist including certain randomly curved manifolds. These surfaces are potentially realizable  not only through quantum simulators but also as two-dimensional materials involving rough substrates or corrugated surfaces. Apart from the obvious relation to Anderson localization, such transient  membranes give rise to other types of unconventional physics where the familiar two-dimensional behavior may be reversed. For example, a Heisenberg magnet on a curved transient manifold would exhibit a finite-temperature long-range order. In contrast, the BKT transition of the $XY$-model would be lost due to non-logarithmic vortex binding forces. Ideas for further work of relevance to solid-state experiment include search for transient curved bilayer systems, which map on interesting geometric structures~\cite{moireG}.

\begin{acknowledgments}
This work was supported by the U.S. Department of Energy, Office of Science, Basic Energy Sciences under Award No.
DE-SC0001911 and the Julian Schwinger Foundation (L.S.). L.S. and V.G. also acknowledge partial support from the Simons Foundation via Collaboration ``Localization of Waves'' at the initial stages of this project.
\end{acknowledgments}

\nocite{bull,LSW1963,ErdosTuran1948,Aistleitner2014,Bilyk2014,RivatTenenbaum2005,Baker1981,kollar2019hyperbolic,boettcher2020quantum,blais2021circuit,lenggenhager2022simulating}

\bibliography{brownian}

\end{document}


\title{Geometric Delocalization in Two Dimensions\\ Supplemental Material}

\author{Laura Shou}
    \affiliation{Joint Quantum Institute, Department of Physics, University of Maryland, College Park 20742}
    
\author{Alireza Parhizkar}
    \affiliation{Joint Quantum Institute, Department of Physics, University of Maryland, College Park 20742}

\author{Victor Galitski}
    \affiliation{Joint Quantum Institute, Department of Physics, University of Maryland, College Park 20742}

\maketitle
\tableofcontents

\section{Proof of the transience condition in Claim 1}

In this section we prove Claim 1 of the main text, which was stated as follows:
\begin{claim}\label{claim}
For a manifold with smooth metric $g_{ab}=\begin{bmatrix}A^2(r,\theta)&B(r,\theta)\\B(r,\theta)&C^2(r,\theta)\end{bmatrix}$ in polar-like coordinates with $g(r,\theta)>0$ for all $r>0$, the condition
\begin{align}\label{eqn:transience-cond}
\operatorname{meas}\left\{\theta\in[0,2\pi):\int_1^\infty\frac{A^2(r,\theta)}{\sqrt{g(r,\theta)}}\,dr<\infty\right\}>0
\end{align}
implies transience.
\end{claim}

In order to prove this, we will use the definition of the \emph{capacity} of a pair of sets $(K,\Omega)$, where $K$ is a compact set in an open set $\Omega$, both in the manifold $M$ which has metric $g$. Note that only here in the proof subsection, we are designating the metric with $g$ and the metric determinant with $\det g$ to avoid index confusion later on.
We provide a brief definition of capacity here, and refer the reader to Refs.~\cite[\S4.3]{bull} and \cite{LSW1963} for further details and references.
The capacity is defined as
\begin{align}\label{eqn:capacity}
\operatorname{Cap}(K,\Omega):=\inf_{\phi\in\mathcal L(K,\Omega)}\int_\Omega|\nabla \phi|_g^2\,d\mathrm{Vol}_g,
\end{align}
where the infimum is over all locally Lipschitz function $\phi$ on $M$ with $0\le\phi\le1$ and $\phi|_K=1$ and $\phi|_{\bar\Omega^c}=0$, and where $|\nabla \phi|_g^2$ is taken with respect to the metric $g$ and $d\mathrm{Vol}_g$ is the volume element.
The space of functions $\mathcal{L}(K,\Omega)$ can have additional restrictions, such as smoothness, without changing the capacity \cite[\S4.3]{bull}.
The infimum in Eq.~\eqref{eqn:capacity} is obtained by a harmonic function $u$, satisfying
\begin{align}\label{eqn:dirichlet}
\begin{cases}\Delta u=0\\u|_{\partial\Omega}=0\\u|_{\partial K}=1\end{cases},
\end{align}
in which case
\begin{align}\label{eqn:cap-flux}
\operatorname{Cap}(K,\Omega)=\int_\Omega|\nabla u|_g^2\,d\mathrm{Vol}_g. 
\end{align}
The capacity can sometimes be interpreted as a conductivity between $\partial K$ and $\partial\Omega$, since by Green's formula, Eq.~\eqref{eqn:cap-flux} can be written as a flux through the boundary $\partial K$ or $\partial\Omega$.

When $\Omega=M$, we can write $\operatorname{Cap}(K)$ for $\operatorname{Cap}(K,\Omega)$, and call $\operatorname{Cap}(K)$ just the capacity of the set $K$.
Considering these capacities is useful in our case due to the following theorem; for references and proof see e.g. Theorem 5.1 in the overview \cite{bull}.
\begin{thm}\label{thm:capacity-bm}
Brownian motion on a Riemannian manifold $M$ is transient if and only if the capacity of some compact set is positive.
\end{thm}

Using this theorem, our goal will be to show that the capacity of a ball is positive when the condition Eq.~\eqref{eqn:transience-cond} of Claim~\ref{claim} holds.
We will follow the general proof method given in Ref.~\cite[\S12]{bull} [which proves Eq.~\mainmeas\ of the main text], with some changes since the metric is not diagonal, and the coordinate radial distance is different than the geodesic or physical distance. 

\begin{proof}[Proof of Claim~\ref{claim}]
Our goal is to lower bound the capacity $\operatorname{Cap}(B_{r_0},B_R)$. In order to do this, we split the ball $B_R$ into many wedges $\{(r,\theta):r\le R,\theta\in\omega_i\}$, and consider the capacity on each small wedge. We will estimate the metric $g$ on each wedge by a rotationally symmetric one $\tilde{g}_i(r)$, since the the metric cannot vary much if $\omega_i\subseteq\R/(2\pi\Z)$ is small. In the rotationally symmetric case, we can solve the Dirichlet problem Eq.~\eqref{eqn:dirichlet} explicitly to estimate the capacity.

For an arc $\omega\subseteq\R/(2\pi\Z)$, let $N^\omega$ be the cone $\{(r,\theta):\theta\in\omega,r>0\}$. 
Consider a smooth rotationally symmetric metric
\begin{align}
\tilde g=\begin{bmatrix}\tilde{A}^2(r)&\tilde B(r)\\\tilde B(r)&\tilde C^2(r)\end{bmatrix}
\end{align}
on the cone $B_R^\omega:=N^\omega\cap B_R$, such that
\begin{align}\label{eqn:metric-r-approx}
 \det g(r,\theta)\ge \det\tilde g(r)>0,\quad\text{and} \quad A(r,\theta)\le \tilde A(r),
\end{align}
for $r_0\le r\le R$ and $\theta\in\omega$.

Before starting the capacity estimates, we first need to bound the gradient norm in terms of its radial part so we can compare to the capacity for a purely radial metric. For a metric $g=\begin{bmatrix}g_{rr}&g_{r\theta}\\g_{r\theta}&g_{\theta\theta}\end{bmatrix}$, we have
\begin{align*}
|\nabla\phi|_g^2&=(\partial_r\phi,\partial_\theta\phi)\cdot g^{-1}(\partial_r\phi,\partial_\theta\phi)\\
&=\frac{1}{\det g}\left[g_{\theta\theta}(\partial_r\phi)^2+g_{rr}(\partial_\theta\phi)^2-2g_{r\theta}(\partial_r\phi\partial_\theta\phi)\right]\\
&\ge\frac{1}{\det g}\left[g_{\theta\theta}(\partial_r\phi)^2+g_{rr}(\partial_\theta\phi)^2-2\left(\frac{\big(\sqrt{g_{rr}}\partial_\theta\phi\big)^2+\big(\frac{g_{r\theta}}{\sqrt{g_{rr}}}\partial_r\phi\big)^2}{2}\right)\right]\\
&\ge\frac{1}{\det g}\left(g_{\theta\theta}-\frac{g_{r\theta}^2}{g_{rr}}\right)(\partial_r\phi)^2
=\frac{1}{g_{rr}}(\partial_r\phi)^2,\numberthis\label{eqn:grad-radial}
\end{align*}
where we used the inequality $ab\le (a^2+b^2)/2$ applied to appropriate $a,b$ in the third line.
Applying Eq.~\eqref{eqn:grad-radial} to the metric $g=\begin{bmatrix}A^2(r,\theta)&B(r,\theta)\\B(r,\theta)&C^2(r,\theta)\end{bmatrix}$, then 
\begin{align*}
\inf_{\phi\in\mathcal L(B_{r_0},B_R)}\int_{B_R^{\omega}}|\nabla\phi|_g^2\,d\mathrm{Vol}_g &\ge \inf_{\phi\in\mathcal L(B_{r_0},B_R)}\int_\omega\int_{r_0}^{R}\frac{1}{A^2(r,\theta)}(\partial_r\phi(r,\theta))^2\sqrt{\det g(r,\theta)}\,dr\,d\theta\\
&\ge \inf_{\phi\in\mathcal L(B_{r_0},B_R)}\int_\omega\int_{r_0}^{R}\frac{1}{\tilde A^2(r)}(\partial_r\phi(r,\theta))^2\sqrt{\det\tilde g(r)}\,dr\,d\theta,\numberthis\label{eqn:nonradial-mid}
\end{align*}
by the definition of the radial metric $\tilde g$.
These two approximations, the radial metric $\tilde g$ and the bound Eq.~\eqref{eqn:grad-radial}, thus remove $\theta$-dependence from the problem.
To handle Eq.~\eqref{eqn:nonradial-mid}, define the functional
\begin{align}
\mathcal E[\phi]&:=\int_\omega\int_{r_0}^R\frac{1}{\tilde A^2(r)}(\partial_r\phi(r,\theta))^2\sqrt{\det\tilde g(r)}\,dr\,d\theta
\end{align}
over $\phi\in\mathcal L(B_{r_0},B_R)$, so $\phi|_{B^\omega_{r_0}}=1$ and $\phi|_{B^\omega_R}=0$. The Euler--Lagrange equation,
which is determined by evaluating $\lim_{\epsilon\to0}\frac{1}{\epsilon}(\mathcal E[\phi+\epsilon v]-\mathcal E[\phi])$ for any Lipschitz $v$ that is zero on the ball boundaries, 
is 
\begin{align}\label{eqn:el}
\partial_r\left(\frac{1}{\tilde A^2(r)}\sqrt{\det\tilde g(r)}\,\partial_r\phi(r,\theta)\right)=0, 
\end{align} 
or equivalently,
\begin{align}
\partial_r^2\phi+\left(\frac{\partial_r\sqrt{\det\tilde g(r)}}{\sqrt{\det\tilde g(r)}}-\frac{\partial_r\tilde A^2(r)}{A^2(r)}\right)\partial_r\phi=0,
\end{align}
on the region $r_0\le r\le R$ and $\theta\in\omega$.
Eq.~\eqref{eqn:el} implies that
\begin{align*}
\partial_r\phi(r,\theta)&=\frac{\tilde A^2(r)}{\sqrt{\det\tilde g(r)}}y(\theta),\quad\text{and so}\;\phi(R,\theta)-\phi(r,\theta)=y(\theta)\int_r^R\frac{\tilde A^2(\rho)}{\sqrt{\det\tilde g(\rho)}}\,d\rho,
\end{align*}
for some (possibly $\theta$-dependent) constant $y(\theta)$.
Imposing the boundary conditions $\phi|_{\partial B_R\cap N^\omega}=0$ and $\phi|_{\partial B_{r_0}\cap N^\omega}=1$, i.e. $\phi(R,\theta)=0$ and $\phi(r_0,\theta)=1$ for $\theta\in\omega$, we obtain the solution to Eq.~\eqref{eqn:el} is
\begin{align}
\phi_0(r,\theta)=\phi_0(r)&=a\int_r^R\frac{\tilde A^2(\rho)}{\sqrt{\det\tilde g(\rho)}}\,d\rho,
\end{align}
where $a=-y(\theta)=\left(\int_{r_0}^R\frac{\tilde A^2(\rho)}{\sqrt{\det\tilde g(\rho)}}\,d\rho\right)^{-1}$.
We also have by direct evaluation,
\begin{align}
\mathcal E[\phi_0]&=|\omega|a=|\omega|\left(\int_{r_0}^R\frac{\tilde A^2(\rho)}{\sqrt{\det\tilde g(\rho)}}\,d\rho\right)^{-1}.
\end{align}
Returning to Eq.~\eqref{eqn:nonradial-mid}, we see that
\begin{align}
\inf_{\phi\in\mathcal L(B_{r_0},B_R)}\int_{B_R^{\omega}}|\nabla\phi|_g^2\,d\mathrm{Vol}_g&\ge |\omega|\left(\int_{r_0}^R\frac{\tilde A^2(\rho)}{\sqrt{\det\tilde g(\rho)}}\,d\rho\right)^{-1}.
\end{align}

If we have a finite set of disjoint arcs $\omega_i\subset\R/(2\pi\Z)$, and metrics $\tilde{g}_i$ satisfying Eq.~\eqref{eqn:metric-r-approx} in $[r_0,R]\times\omega_i$, then
\begin{align*}
\operatorname{Cap}(B_{r_0},B_R)&=\inf_{\phi\in\mathcal L(B_r,B_R)}\int_{B_R}|\nabla \phi|_g^2\,d\mathrm{Vol}_g\\
&\ge\sum_i\inf_{\phi\in\mathcal L(B_r,B_R)}\int_{B_R^{\omega_i}}|\nabla \phi|_g^2\,d\mathrm{Vol}_g\\
&\ge \sum_i|\omega_i|\left(\int_{r_0}^R\frac{\tilde A_i^2(\rho)}{\sqrt{\det\tilde g_i(\rho)}}\,d\rho\right)^{-1}.\numberthis
\end{align*}
By partitioning $[0,2\pi)$ into small arcs $\omega_i$, and taking $\tilde{A}_i(r)$ and $\det\tilde{g}_i(r)$ close to $A(r,\theta)$ and $\det g(r,\theta)$, uniformly in $\theta\in\omega_i$ and $r\in[r_0,R]$, the above sum approximates the integral $\int_0^{2\pi}\left(\int_{r_0}^R\frac{A^2(r,\theta)}{\sqrt{\det g(r,\theta)}}\,d\rho\right)^{-1}\,d\theta$. We thus obtain
\begin{align}
\operatorname{Cap}(B_{r_0},B_R)&\ge\int_0^{2\pi}\left(\int_{r_0}^R\frac{A^2(r,\theta)}{\sqrt{\det g(r,\theta)}}\,d\rho\right)^{-1}\,d\theta.
\end{align}
Taking $R\to\infty$ and applying Theorem~\ref{thm:capacity-bm} yields the claim.
\end{proof}

\section{Transience details for the tablecloth manifold}

In this section, we give precise conditions for the bump functions $\chi_n$ used in the main text, and then give the details showing that the tablecloth membrane given in Eq.~\maintable\ of the main text is transient via the condition in Eq.~\eqref{eqn:transience-cond}.

First, let $\{\eta_L(r)\}_{L\in\N}$ be a collection of smooth bump functions, such that $\eta_L\equiv 1$ on $[1/4,L-1/4]$ and $\eta_L\equiv0$ outside $[1/8,L-1/8]$, and $\sup_{L\in\N}\|\partial_r\eta_L\|_\infty<\infty$. Then set $\chi_n(r):=\eta_{n}(r-n(n-1)/2)$.
One can take for example
\begin{align*}
\eta_L(r)&=\begin{cases}
0,&r\le\frac{1}{8}\text{ or }r\ge L-\frac{1}{8}\\
\exp\left({1-\frac{1}{1-64(r-1/4)^2}}\right),&\frac{1}{8}\le r\le\frac{1}{4}\\
1,&\frac{1}{4}\le r\le L-\frac{1}{4}\\
\exp\left({1-\frac{1}{1-64(r-L+1/4)^2}}\right),&L-\frac{1}{4}\le r\le L-\frac{1}{8}
\end{cases}.
\end{align*}

Recall that Eq.~\maintable\ of the main text defines a tablecloth membrane via the height function,
\begin{align}\label{eqn:h}
h(r,\theta)&=\sum_{n=1}^\infty \chi_n(r)\cos(n^4\theta) \, ,
\end{align}
where $\chi_n(r)$ satisfy the general requirements described above.
Now for transience, recall that for a metric $ds^2=A^2(r,\theta)\,dr^2+2B(r,\theta)\,dr\,d\theta+C^2(r,\theta)\,d\theta^2$ with determinant denoted by $g(r,\theta)$, that we want to show
\begin{align}\label{eqn:transience-cond2}
\int_1^\infty \frac{A^2(r,\theta)}{\sqrt{g(r,\theta)}}\,dr<\infty,
\end{align}
for a positive measure set of $\theta$.
From Eq.~\maingI\ of the main text, the induced metric of the tablecloth membrane in Eq.~\eqref{eqn:h} is, at $(r,\theta)$ with $n(n-1)/2<r<n(n+1)/2$,
\begin{align}
g_{ab}&=\begin{bmatrix}
1+ \chi_n'(r)^2\cos^2(n^4\theta)& -n^4\chi_n(r)\chi_n'(r)\cos(n^4\theta)\sin(n^4\theta)\\
-n^4\chi_n(r)\chi_n'(r)\cos(n^4\theta)\sin(n^4\theta)&r^2+n^8\chi_n(r)^2\sin^2(n^4\theta)
\end{bmatrix}.
\end{align}
For $\frac{n(n-1)}{2}+\frac{1}{4}<r<\frac{n(n+1)}{2}-\frac{1}{4}$, the metric simplifies because $\chi_n\equiv1$ there, leading to the much simpler formula
\begin{align}
g_{ab}&=\begin{bmatrix}
1& 0\\
0&r^2+n^8\sin^2(n^4\theta)
\end{bmatrix}.
\end{align}
Splitting up the integral in Eq.~\eqref{eqn:transience-cond2} into regions where $\chi_n\equiv1$ or not, and using that ${1+|\chi_n'(r)|^2}\le c'$ for some constant $c'$ and all $n$ and $r$ (by construction of $\chi_n$), we can estimate
\begin{align*}
\int_1^\infty\frac{A^2(r,\theta)}{\sqrt{g(r,\theta)}}\,dr&\le \sum_{n=1}^\infty\Bigg(\int_{\frac{n(n-1)}{2}+\frac{1}{4}}^{\frac{n(n+1)}{2}-\frac{1}{4}}\frac{1}{\sqrt{r^2+n^8\sin^2(n^4\theta)}}\,dr\\
&\hspace{4cm}+\int_{\frac{n(n+1)}{2}-\frac{1}{4}}^{\frac{n(n+1)}{2}+\frac{1}{4}}\frac{1+\chi_n'(r)^2\cos^2(n^4\theta)}{\sqrt{r^2+n^8\chi_n(r)^2\sin^2(n^4\theta)+r^2(\chi_n'(r))^2\cos^2(n^4\theta)}}\,dr\Bigg)\\
&\le \sum_{n=1}^\infty\Bigg(\frac{n}{\sqrt{\frac{n^2(n-1)^2}{4}+n^8\sin^2(n^4\theta)}}+\frac{c'}{n(n+1)-\frac{1}{2}}\Bigg),\numberthis\label{eqn:integral-series}
\end{align*}
where to obtain the second term in the last line we used that the corresponding integrand is bounded from above by $\frac{c'}{r}\le \frac{c'}{n(n-1)/2+1/4}$.
This leads to the sum $\sum_{n=1}^\infty\frac{c'}{n(n+1)-1/2}$, which is $\le\sum_{n=1}^\infty\frac{c'}{n^2}<\infty$. 
Considering then the first term in Eq.~\eqref{eqn:integral-series}, we need to show that for a positive measure set of $\theta$, that
\begin{align}\label{eqn:series}
\sum_{n=1}^\infty\frac{1}{n\sqrt{1+n^4\sin^2(n^4\theta)}}<\infty.
\end{align}
The terms in the series tend to decay like $1/n^3$, except when $\sin^2(n^4\theta)$ is small, which is when $n^4\theta$ is close to a multiple of $\pi$. For these $n$ the terms behave like $1/n$.
We can split the series based on this into two parts; letting $\{\cdot\}$ denote the value mod 1 of a number, taken to be in $(-1/2,1/2]$ (essentially the fractional part of a number but using negative numbers for convenience), 
then fixing a small $\epsilon>0$, 
\begin{align*}
\sum_{n=1}^\infty\frac{1}{n\sqrt{1+n^4\sin^2(n^4\theta)}}&= \sum_{n=1}^\infty \frac{\mathbf{1}_{|\{n^4\theta/\pi\}|\ge n^{-(2-\epsilon)}}}{n\sqrt{1+n^4\sin^2(n^4\theta)}} +
 \sum_{n=1}^\infty \frac{\mathbf{1}_{|\{n^4\theta/\pi\}|<n^{-(2-\epsilon)}}}{n\sqrt{1+n^4\sin^2(n^4\theta)}}\\
 &\le \sum_{n=1}^\infty\frac{1}{n\sqrt{1+cn^{2\epsilon}}}+\sum_{n=1}^\infty\frac{1}{n}\mathbf{1}_{|\{n^4\theta/\pi\}|<n^{-(2-\epsilon)}},
 \numberthis\label{eqn:series-split}
\end{align*}
where for the first term we used that $\sin^2(x)\ge c|\{x/\pi\}|^2$. 
The first series in Eq.~\eqref{eqn:series-split} converges, so we just need to show convergence of the second term involving the harmonic series summed only over certain $n$.
This convergence is intuitively plausible, since $n^{-(2-\epsilon)}$ is shrinking rapidly, so the $n\in\N$ contributing to the sum must be fairly sparse.
By Weyl's equidistribution theorem, we know that $\{n^4x\}$ is equidistributed in $\R/\Z$ for almost every (a.e.) $x$, i.e.,
\begin{align*}
\frac{\#\{1\le n\le N:\{n^4 x\}\in[a,b]\}}{N}\xrightarrow{N\to\infty} \operatorname{length}([a,b]),
\end{align*}
for any interval $[a,b]\subseteq\R/\Z$.
This suggests that if the equidistribution is fast enough to handle the shrinking interval, the occurrence $|\{n^4 x\}|<n^{-(2-\epsilon)}$ should only happen approximately around a fraction maybe $2n^{-(2-\epsilon)}$ of the time, which should lead to a convergent sum. 
In order to prove the convergence, we can apply the quantitative estimate on the rate of equidistribution given by the Erd\H{o}s--Tur\'an inequality \cite{ErdosTuran1948}, which states that for any $N,m\in\N$ and $y_n\in[0,1)$ and $I\subseteq[0,1)$,
\begin{align}
\left|\#\{1\le n\le N:y_n\in I\}-N|I|\right|&\le \frac{3N}{m}+3\sum_{k=1}^m\frac{1}{k}\left|\sum_{n=1}^N e^{2\pi ik y_n}\right|.
\end{align}
(For further details, see e.g. Refs.~\cite{Aistleitner2014,Bilyk2014}, and for better constants, see Ref.~\cite{RivatTenenbaum2005}.)
Applying this with Carlson's theorem on pointwise a.e. convergence of Fourier series (or more precisely the maximal operator bound) to bound the exponential sum, Ref.~\cite{Baker1981} showed that for any strictly increasing sequence of natural numbers $a_1,a_2,\ldots$, there is a constant $C$ (which does not depend on $N$) so that for every $\delta>0$,
\begin{align}\label{eqn:baker}
\sup_{I\subseteq\R/\Z}\left|\#\{1\le n\le N:\{a_nx\}\in I\}-N|I|\right|&\le C N^{1/2}(\log N)^{3/2+\delta},
\end{align}
for a.e. $x$ with respect to the Lebesgue measure. 
In our case, we will take $a_n=n^4$ and $x=\theta/\pi$.
First, we rewrite the second sum in Eq.~\eqref{eqn:series-split} so that we can apply Eq.~\eqref{eqn:baker}. For an increasing function $f:\N\to\N$ with $f(1)=1$, we can group the terms in the series to write
\begin{align}\label{eqn:f-bound}
\sum_{n=1}^\infty\frac{1}{n}\mathbf{1}_{|\{n^4\theta/\pi\}|<n^{-(2-\epsilon)}}&\le \sum_{k=1}^\infty\frac{1}{f(k)} \#\{f(k)\le n\le f(k+1):|\{n^4 x\}|\le n^{-(2-\epsilon)}\}.
\end{align}
Because we have such fast decay $n^{-(2-\epsilon)}$, which causes the $n$ in the sum to be much sparser than we actually need, we do not need to be very careful with estimates. Applying Eq.~\eqref{eqn:baker}, we can simply estimate
\begin{align*}
\#\{f(k)\le n\le f(k+1):|\{n^4 x\}|\le n^{-(2-\epsilon)}\}&\le \#\{1\le n\le f(k+1):|\{n^4 x\}|\le f(k)^{-(2-\epsilon)}\}\\
&\le 2f(k)^{-(2-\epsilon)}f(k+1)+Cf(k+1)^{1/2}(\log f(k+1))^2.\numberthis
\end{align*}
Taking $f(k)=k^3$, combining this with Eq.~\eqref{eqn:f-bound} gives
\begin{align}
\sum_{n=1}^\infty\frac{1}{n}\mathbf{1}_{|\{n^4\theta/\pi\}|<n^{-(2-\epsilon)}}&\le C'\sum_{k=1}^\infty\left(\frac{1}{k^{6-3\epsilon}}+\frac{(\log k)^2}{k^{3/2}}\right)<\infty,
\end{align}
so that Eq.~\eqref{eqn:series-split} is finite for a.e. $\theta$. Thus Eq.~\eqref{eqn:integral-series} is finite for a.e. $\theta$, implying the membrane $h$ is transient by Claim~\ref{claim}. \qed

\section{Circuit board examples}

Recall the condition for transience of rotationally symmetric manifolds with metric of the form $ds^2=dr^2+f(r)^2\,d\theta^2$. This was stated in Eq.~\mainrot\ as the condition
\begin{align}\label{eqn:rot-transience-sup}
    \int_1^\infty\frac{1}{f(r)}\,dr<\infty.
\end{align}
Using this, we can quickly write down several examples of rotationally symmetric metrics $ds^2=dr^2+f(r)^2\,d\theta^2$ with transient Brownian motion and average zero curvature.
However, as shown in the main text, such manifolds \emph{cannot} be realized as membranes in flat 3D space.
However, they can still be realized using circuit boards \cite{kollar2019hyperbolic,boettcher2020quantum,blais2021circuit,lenggenhager2022simulating} by approximating the manifold by a graph, where the number of neighbors at each $r=R$ is given by $2\pi f(R)$. 

\vspace{2mm}
\noindent\textit{Example} 1 (Circuit example, non-decaying curvature).
Let $f(r)=r^{1+\varepsilon}(2+\cos r)+r$ for some $0<\varepsilon\le 1$, which by Eq.~\eqref{eqn:rot-transience-sup} ensures transience. The extra $+r$ at the end makes $f'(0)=1$, which gives some regularity of the manifold at $0$.
To check the average zero curvature condition, we calculate the Gaussian curvature to be 
\begin{align}\label{eqn:K-1}
K(r)&=-\frac{f''(r)}{f(r)}=\frac{-\varepsilon(1+\varepsilon)(2+\cos r)r^{\varepsilon-1}+2(1+\varepsilon)r^\varepsilon \sin r+r^{1+\varepsilon}\cos r}{r^{1+\varepsilon}(2+\cos r)+r}.
\end{align}
The behavior as $r\to\infty$ looks like $\frac{\cos r}{2+\cos(r)}+o(1)$ which oscillates and doesn't decay.
For ``average curvature zero'', we will show there is a sequence of $R_k\to\infty$ with
\begin{align}\label{eqn:intzero}
\int_{B(0,R_k)}K(r,\theta)\,d\mathrm{Vol}(r,\theta)=0.
\end{align} 
We can evaluate
\begin{align}
I(R):=\int_{B(0,R)}K(r)\,d\mathrm{Vol}(r,\theta)
=-2\pi\int_0^R\frac{f''(r)}{f(r)}f(r)\,dr
&=-2\pi [f'(R)-1]\\
&=-2\pi R^{\varepsilon}\left[(1+\varepsilon)\cos R-R\sin R+2\varepsilon+2\right].
\end{align}
The leading order term in $I(R)$ is $2\pi R^{1+\varepsilon}\sin R$, which is highly oscillatory and ensures the existence of the $R_k$ in Eq.~\eqref{eqn:intzero}.

We expect this manifold satisfies an alternative ``average zero curvature condition'',
\begin{align}\label{eqn:k-half}
\lim_{R\to\infty}\frac{\int_{B(0,R)}(K(r,\theta))_-\,d\mathrm{Vol}(r,\theta)}{\int_{B(0,R)}|K(r,\theta)|\,d\mathrm{Vol}(r,\theta)}
&=\lim_{R\to\infty}\frac{\int_{B(0,R)}(K(r,\theta))_+\,d\mathrm{Vol}(r,\theta)}{\int_{B(0,R)}|K(r,\theta)|\,d\mathrm{Vol}(r,\theta)},
\end{align}
where $(K(r,\theta))_-:=|K(r,\theta)|\mathbf{1}_{\{K(r,\theta)<0\}}$ and $(K(r,\theta))_+:=K(r,\theta)\mathbf{1}_{\{K(r,\theta)>0\}}$. 
In other words, this condition says the amount of positive and negative curvature each take up half the total curvature in the limit; or, the positive and negative curvature integrals have the same leading order behavior.
Note the condition in Eq.~\eqref{eqn:k-half} is different than just requiring the limit of the average curvature to be zero, which would be $\langle K\rangle\equiv \lim_{R\to\infty}\frac{1}{|B(0,R)|}\int_{B(0,R)}K(r)\,d\mathrm{Vol}(r,\theta)=0$. This latter condition can hold even with everywhere negative curvature that tends to zero quickly at infinity.

\vspace{2mm}
\noindent\textit{Example} 2 (Circuit example, decaying curvature).
Let $f(r)=r^{1+\varepsilon}+r\cos(r)$ for some $\varepsilon>0$.
Then $f(0)=0$, $f'(0)=1$, and
\begin{align*}
K(r)&=\frac{-\varepsilon(1+\varepsilon)r^{\varepsilon-1}+2\sin r+r\cos r}{r^{1+\varepsilon}+r\cos(r)},
\end{align*}
which tends to zero at infinity.
We have
\begin{align}
I(R) :=\int_{B(0,R)}K(r)\,d\mathrm{Vol}(r,\theta)
&=-2\pi [f'(R)-1]\\
&=-2\pi(1+\varepsilon)R^\varepsilon+2\pi R\sin R-2\pi \cos R+2\pi.
\end{align}
The leading order term is $2\pi R\sin R$ which is highly oscillatory, and ensures existence of the $R_k$ in Eq.~\eqref{eqn:intzero}.

\bibliography{brownian.bib}